\documentclass[sigconf,screen]{acmart}

\usepackage{graphicx}
\usepackage{booktabs}
\usepackage{xurl}

\usepackage[linesnumbered,ruled,vlined]{algorithm2e}

\usepackage{listings}
\usepackage{xcolor}
\usepackage{textcomp}

\usepackage{array}
\usepackage{tabularx}

\usepackage{amsfonts}
\usepackage{amsthm}

\usepackage{xspace}
\usepackage{microtype}

\usepackage[capitalize]{cleveref}
\crefname{section}{Sect.}{Sects.}
\Crefname{section}{Section}{Sections}
\crefname{definition}{Def.}{Defs.}
\Crefname{definition}{Definition}{Definitions}
\crefname{algorithm}{Alg.}{Algs.}
\Crefname{algorithm}{Algorithm}{Algorithms}
\crefname{theorem}{Thm.}{Thms.}
\Crefname{theorem}{Theorem}{Theorems}

\newtheoremstyle{compact}%
  {2pt}
  {2pt}
  {\itshape}
  {}
  {\bfseries}
  {.}
  {.5em}
  {}
\theoremstyle{compact}
\newtheorem{definition}{Definition}
\newtheorem{theorem}{Theorem}

\definecolor{verylightgray}{rgb}{.97,.97,.97}
\definecolor{tvmblue}{rgb}{0.15,0.30,0.65}
\definecolor{tvmgreen}{rgb}{0.15,0.55,0.25}
\definecolor{tvmred}{rgb}{0.70,0.15,0.15}
\definecolor{tvmgray}{rgb}{0.50,0.50,0.50}

\lstdefinelanguage{tvmasm}{
  keywords=[1]{PUSHCONT,PUSHINT,PUSHCTR,POPCTR,SAVE,SAVEALT,SAVEBOTH,
    EXECUTE,CALLX,JMPX,CALLXARGS,JMPXARGS,CALLCC,
    RET,RETALT,IF,IFNOT,IFELSE,IFJMP,IFNOTJMP,
    SENDRAWMSG,SENDMSG,ACCEPT,SETGASLIMIT,RAWRESERVE,
    LDMSGADDR,LDU,LDI,LDSLICE,CTOS,ENDS,LDGRAMS,
    DUP,SWAP,DROP,XCHG,ROT,ROLL,BLKSWAP,
    ADD,SUB,MUL,DIV,CMP,EQUAL,LESS,GREATER,
    THROWIF,THROWIFNOT,THROW,
    MYADDR,BALANCE,NOW,BLOCKLT,LTIME,
    DICTGET,DICTIGET,CALLDICT,JMPDICT,CALLREF,JMPREF,
    BLESS,SETCONT,SETRETCTR,SETALTCTR,
    RANDU256,RAND,ADDRAND,SETRAND,RANDOMIZE},
  keywordstyle=[1]\color{tvmblue}\bfseries,
  keywords=[2]{c0,c1,c2,c3,c4,c5,c7,s0,s1,s2,s3,s4},  
  keywordstyle=[2]\color{tvmred}\bfseries,
  identifierstyle=\color{black},
  sensitive=true,
  comment=[l]{;;},
  morecomment=[s]{/*}{*/},
  commentstyle=\color{tvmgray}\itshape,
  stringstyle=\color{tvmred}\ttfamily,
  morestring=[b]",
}

\lstdefinelanguage{func}{
  keywords=[1]{int,cell,slice,builder,cont,tuple,return,var,if,ifnot,
    elseif,else,while,do,until,repeat,forall,
    recv_internal,recv_external,run_ticktock,
    method_id,impure,inline,inline_ref,asm,
    global,const},
  keywordstyle=[1]\color{tvmblue}\bfseries,
  keywords=[2]{begin_parse,end_parse,load_uint,load_int,load_msg_addr,
    load_coins,load_ref,store_uint,store_int,store_slice,store_ref,
    send_raw_message,accept_message,set_gas_limit,raw_reserve,
    get_data,set_data,begin_cell,end_cell,
    throw,throw_if,throw_unless},
  keywordstyle=[2]\color{tvmgreen}\bfseries,
  identifierstyle=\color{black},
  sensitive=true,
  comment=[l]{;;},
  morecomment=[s]{/*}{*/},
  commentstyle=\color{tvmgray}\itshape,
  stringstyle=\color{tvmred}\ttfamily,
  morestring=[b]",
}

\newcommand{\tool}{\textsc{TasmScan}\xspace}
\newcommand{\tsa}{\textsc{TSA}\xspace}
\newcommand{\savelist}{\textit{savelist}\xspace}
\newcommand{\tvm}{\textsc{TVM}\xspace}
\newcommand{\ton}{\textsc{TON}\xspace}
\newcommand{\evm}{\textsc{EVM}\xspace}
\newcommand{\tasir}{\textsc{TASIR}\xspace}
\newcommand{\eg}{\textit{e.g.}\xspace}
\newcommand{\etal}{\textit{et al.}\xspace}

\setcopyright{cc}
\setcctype{by-nc-nd}
\acmDOI{10.1145/3832783.3834346}
\acmYear{2026}
\copyrightyear{2026}
\acmISBN{979-8-4007-2882-2/2026/10}
\acmConference[ASE '26]{Proceedings of the 41st IEEE/ACM International Conference on Automated Software Engineering}{October 12--16, 2026}{Munich, Germany}
\acmBooktitle{Proceedings of the 41st IEEE/ACM International Conference on Automated Software Engineering (ASE '26), October 12--16, 2026, Munich, Germany}
\acmSubmissionID{ase26main-p182-p}
\received{2026-03-26}
\received[accepted]{2026-06-18}

\begin{document}

\title[TasmScan: Continuation-Aware Taint Analysis\ldots]{TasmScan:
Continuation-Aware Taint Analysis for TVM Bytecode with Savelist Abstraction}

\author{Yixuan Liu}
\orcid{0009-0006-2255-7901}
\affiliation{%
  \institution{Nanyang Technological University}
  \city{Singapore}
  \country{Singapore}
}
\email{LIUY0255@e.ntu.edu.sg}

\author{Yin Wu}
\orcid{0009-0001-6583-3703}
\affiliation{%
  \institution{Xi'an Jiaotong University}
  \city{Xi'an}
  \country{China}
}
\email{wuyin@stu.xjtu.edu.cn}

\author{Yi Li}
\correspondingauthor
\orcid{0000-0003-4562-8208}
\affiliation{%
  \institution{Nanyang Technological University}
  \city{Singapore}
  \country{Singapore}
}
\email{yi\_li@ntu.edu.sg}


\begin{abstract}
The Open Network (\ton), with a peak market capitalization exceeding \$20 billion
and over 175 million activated on-chain addresses, relies on the \tvm (\ton
Virtual Machine) to execute smart contracts. \tvm uses
first-class continuations with \savelist{}s to manage control flow and register
state across continuation invocations. Since \savelist-captured registers allow data to flow across continuation
boundaries without passing through the operand stack, bytecode-level analyses
cannot construct complete data flow tracking without explicitly modeling
\savelist semantics.
We present \tool, the first bytecode-level static analysis framework for \tvm
that enables cross-continuation data flow reasoning without requiring source
code. \tool models \savelist semantics via forward register analysis with a
formal over-approximation guarantee for exact-resolved save sites and
locally tracked register definitions, then lifts bytecode into \tasir, a typed
intermediate representation, and performs path-sensitive taint analysis with
context-aware sources to detect defects.
We evaluate \tool on 2{,}921 contracts from the \ton verifier registry and a
labeled benchmark of 208 contracts with human-confirmed ground truth. On the
full corpus, \tool resolves 294{,}546 dynamic continuation targets
with 100\% precision; ablation confirms that \savelist propagation is
essential for resolving indirect register calls that depend on
cross-continuation register passing.
On the benchmark, \tool detects 95.3\% of defects across five classes with
96.8\% precision. A 366-pair stratified sample from the full corpus estimates 85.8\%
overall precision. \tool offers a 17$\times$ median speedup over the
state-of-the-art symbolic-execution baseline, and in the path-analysis
comparison completes 100\% of analyses with zero crashes or timeouts.
\end{abstract}



\maketitle


\section{Introduction}\label{sec:intro}


Smart contracts deployed on blockchain platforms manage substantial financial
assets, making their security a critical concern~\cite{atzei2017survey,
perez2021smart}. The Open Network (\ton), initially designed by Nikolai
Durov~\cite{durov2017ton} to support Telegram's large user base, has grown into
a major decentralized platform with a peak market capitalization exceeding
\$20 billion and over 175 million activated on-chain
addresses~\cite{tonscan2026stats}. Its DeFi ecosystem has reached a total value
locked (TVL) of up to \$800 million~\cite{defillama2026ton}, hosting thousands
of smart contracts for token transfers, lending protocols, NFT
marketplaces, and governance mechanisms~\cite{tonblockchain2023}. A recent
empirical study found that approximately 94\% of \ton contracts
contain at least one
defect~\cite{song2025tonscanner}, underscoring the need for automated analysis
to safeguard the significant assets under management.


\begin{figure}[t]
\centering
\includegraphics[width=\columnwidth]{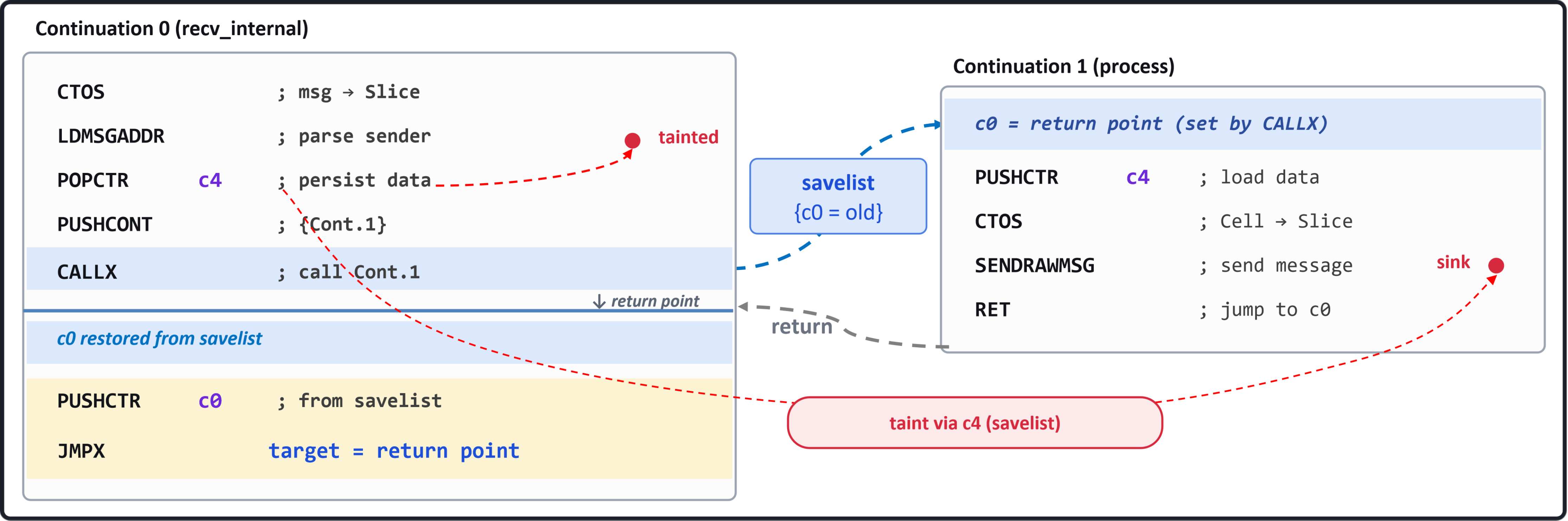}
\caption{Savelist-dependent indirect jump and cross-continuation taint flow.}
\Description{A control-flow sketch showing how a callee stores a continuation in
register c0, returns, and later restores that saved continuation so an indirect
jump target depends on savelist state across continuation boundaries. A red
dashed path shows taint flowing from the sender address read through register c4
via the savelist to a message-send sink in the callee continuation.}
\label{fig:motivating}
\end{figure}

\ton smart contracts are compiled to bytecode and executed on the \tvm.
Unlike the Ethereum Virtual Machine (\evm), which uses address-based jumps,
\tvm employs \emph{first-class continuations} with \savelist-mediated register
save/restore as its sole control-flow primitive~\cite{tvmspec2023}.
\Cref{fig:motivating} illustrates this mechanism: \texttt{CALLX} (blue)
implicitly saves register \texttt{c0} into the callee's \savelist, and a
later \texttt{JMPX} (amber) jumps to the restored value. Without \savelist
modeling, the jump target is unresolved, leaving a gap in the control-flow
graph (CFG). The red dashed path further shows how taint propagates across
continuation boundaries through the same \savelist mechanism.


Analyzing \ton smart contracts poses three challenges.
\emph{First}, source code is often unavailable for deployed contracts, so the
analysis must operate directly on bytecode; TONScanner~\cite{song2025tonscanner}
provides source-level analysis but cannot reach bytecode-level \savelist
behaviors.
\emph{Second}, data flows across continuation boundaries through register
save/restore operations, a pattern invisible to analyses designed for
address-based control flow~\cite{feist2019slither, mueller2018mythril,
tsankov2018securify, bose2022sailfish}.
\emph{Third}, statically abstracting \savelist effects across all execution
paths is difficult; \tsa{}~\cite{esprito2024tsa} faithfully executes \savelist
operations per path via symbolic execution but models each \savelist as a
concrete structure with no join or widening, leaving cross-continuation data
flow unabstracted.


To address these challenges, we propose \tool, which directly analyzes
deployed \tvm bytecode. \tool first resolves continuation targets via stack
simulation to recover the CFG, then applies forward register analysis to
statically reconstruct each continuation's \savelist effects, with a soundness
guarantee scoped to exact-resolved save sites and locally tracked register
definitions. Built on this foundation,
\tool lifts bytecode into a typed IR (\tasir) and performs continuation-aware
taint analysis to detect 5 defect classes. Evaluation on 2{,}921 real-world
contracts shows that \tool is both effective and efficient compared with the
state-of-the-art symbolic-execution baseline.


In summary, this paper makes the following contributions:
\begin{itemize}
  \item \textbf{Formal savelist abstraction.}
    We formulate a set-valued dataflow model for \tvm \savelist semantics
    with a tracked-definition soundness theorem under exact-resolved save
    sites, enabling cross-continuation data flow reasoning at the bytecode
    level while making the proof scope explicit.
    (\cref{sec:savelist})

  \item \textbf{\tasir and continuation-aware taint analysis.}
    We design \tasir, a typed intermediate representation for \tvm bytecode,
    and implement continuation-aware taint analysis on top of it to detect
    5 defect classes in \tvm/\ton smart contracts.
    (\cref{sec:ir,sec:taint,sec:detectors})

  \item \textbf{Comprehensive evaluation.}
    We evaluate \tool on 2{,}921 contracts from the \ton verifier registry.
    \tool resolves 294{,}546 dynamic continuation targets with 100\%
    precision, including 1{,}028 \savelist-dependent targets across
    129 contracts that require cross-continuation register tracking.
    It achieves a 95.3\% detection rate with a 17$\times$ median speedup
    over \tsa{} and 100\% analysis completion with zero crashes.
    (\cref{sec:evaluation})
\end{itemize}


\section{Background}\label{sec:background}

\subsection{The TON Blockchain}\label{sec:bg-ton}

The \ton blockchain~\cite{durov2017ton, tonblockchain2023} is a multi-chain
platform using asynchronous message passing between smart contracts.
Developers write contracts in \textit{FunC} (a C-like language),
\textit{Tact} (a higher-level language), \textit{Tolk} (a TypeScript-like
language), or directly in \textit{Fift} (a low-level stack language); all
compile to \tvm bytecode, serialized in the \emph{Bag of Cells} (BOC) format
and deployed on-chain. Since source code is not always
available for deployed contracts, bytecode-level analysis is essential.

\ton contracts define standard entry points dispatched by a selector value on
the stack. The primary entry point, \texttt{recv\_internal} (selector~0),
handles messages from other contracts. The initial stack provides the
selector (\texttt{s0}), message body (\texttt{s1}), and message
cell (\texttt{s2}) as inputs to the contract logic.

\subsection{TVM and Continuations}\label{sec:bg-cont}

The \evm uses program-counter-driven control transfer
(\texttt{JUMP}\slash\texttt{JUMPI})~\cite{ethereumyellowpaper}.
The \tvm~\cite{tvmspec2023} replaces this model with \emph{first-class
continuations}: each continuation is a composite value containing executable
code and a \savelist, a partial mapping from register indices to values.
This mechanism serves as the sole control-flow primitive in \tvm.

The \tvm is a stack-based virtual machine with 16 control register slots
(\texttt{c0}--\texttt{c15}), of which seven have defined roles, and a
tree-structured data model based on \emph{Cells} (up to 1,023 bits of data
and 4 references to other Cells), \emph{Slices} (read cursors), and
\emph{Builders} (write cursors). Four control registers are central to our analysis:
\texttt{c0} holds the return continuation,
\texttt{c1} the alternative-return continuation,
\texttt{c3} the method dictionary, and
\texttt{c4} stores persistent contract data.
The remaining slots serve roles such as exception handling (\texttt{c2}),
output action accumulation (\texttt{c5}), and blockchain context
(\texttt{c7})~\cite{tvmspec2023}. We write
$\mathit{Reg} = \{c_0, c_1, c_2, c_3, c_4, c_5, c_7\}$ for the set of
active control registers.

When control transfers between continuations, register values can be
preserved through a mechanism called the \savelist:

\begin{definition}[\tvm Continuation]\label{def:continuation}
  An ordinary \tvm continuation is a tuple
  $c = \langle \mathit{code}, \ell \rangle$
  where $\mathit{code}$ is an instruction sequence and
  $\ell : \mathit{Reg} \rightharpoonup \mathit{Value}$ is a partial
  function mapping register indices to values (the \savelist).%
  \footnote{Runtime continuations carry additional fields (stack, codepage,
  argument count)~\cite{tvmspec2023}; the code and \savelist are the fields
  relevant to cross-continuation data flow.}
\end{definition}

The \texttt{SAVE}~$r$ instruction writes the current value of register~$r$ into
the \savelist of the return continuation (\texttt{c0}). When a continuation is
invoked (\eg, via \texttt{EXECUTE}, \texttt{CALLX}, \texttt{JMPX},
or conditional instructions), its \savelist entries are restored: for each
register index $r$ in the domain of $\ell$ (i.e., each register that was
explicitly saved), $r$ is set to $\ell(r)$. This mechanism enables data
to flow across continuation boundaries through register save/restore
operations, a pattern that requires dedicated static analysis support.

\Cref{tab:evm-tvm} summarizes the key architectural differences between
\evm and \tvm that motivate the design of \tool.

\begin{table}[t]
  \caption{Architectural comparison between EVM and TVM.}
  \label{tab:evm-tvm}
  \centering
  \begin{tabularx}{\columnwidth}{@{}p{0.30\columnwidth}XX@{}}
    \toprule
    \textbf{Aspect} & \textbf{EVM} & \textbf{TVM} \\
    \midrule
    Control flow    & JUMP/JUMPDEST   & First-class continuations \\
    State transfer  & None            & Savelist (register snapshot) \\
    Data model      & 256-bit words   & Cell/Slice/Builder tree \\
    Message format  & ABI encoding    & TL-B encoding \\
    Function dispatch & 4-byte selector & Dictionary dispatch \\
    Gas model       & Prepaid         & Credit-based gas \\
    \bottomrule
  \end{tabularx}
\end{table}


\section{Approach}\label{sec:approach}

\begin{figure*}[t]

\centering
\includegraphics[width=\textwidth]{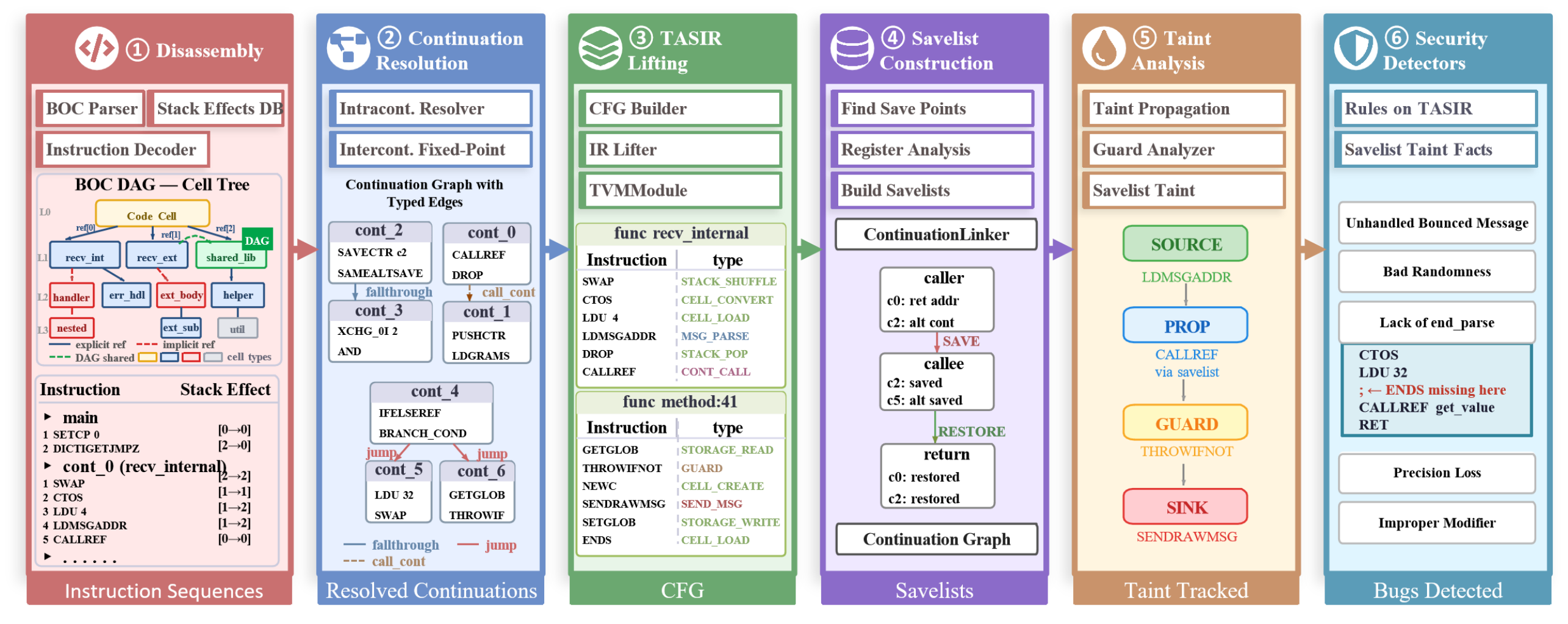}
\caption{Overview of \tool.}
\Description{Pipeline diagram with six stages: bytecode disassembly, continuation
resolution, lifting to TASIR, static savelist modeling, path-sensitive taint
analysis, and security detectors. Data products flow between stages as decoded
instructions, resolved continuation facts, typed IR, savelist summaries, and
taint findings.}
\label{fig:arch}
\end{figure*}

\Cref{fig:arch} shows the \tool pipeline.
BOC disassembly (\cref{sec:disasm}) decodes BOC-serialized bytecode
into typed instructions with stack-effect annotations.
Continuation resolution (\cref{sec:cont-resolve}) performs stack
simulation to resolve dynamic continuation targets and build the CFG.
TASIR lifting (\cref{sec:ir}) translates resolved instructions into
a typed intermediate representation.
Savelist construction (\cref{sec:savelist}) statically models
register save/restore effects across continuation boundaries.
Taint analysis (\cref{sec:taint}) is path sensitive and propagates taint across
continuation boundaries via savelists.
Security detectors (\cref{sec:detectors}) identify defects from
taint findings and structural patterns.

\subsection{BOC Disassembly}\label{sec:disasm}

The disassembly stage converts BOC-serialized \tvm bytecode into typed
instructions annotated with stack effects.

\textit{Instruction database.}
Our disassembler's opcode database covers the \tvm instruction set with encoding,
mnemonic, and interval stack-effect annotations, derived from the official
specification~\cite{tvmspec2023,tonblockchain2023}.

\textit{Variable-length decoding.}
\tvm opcodes are variable-length (8--24~bits). The decoder performs prefix
matching over sorted opcode ranges: it aligns the remaining slice to a 24-bit
window, identifies the next opcode, and extracts its operands.

\textit{All-or-nothing fallback.}
When decoding fails mid-slice (\eg, in data cells), the decoder discards the
partial result and emits the entire slice as raw data, avoiding spurious
instructions. The raw payload remains attached to the BOC cell tree but is not
treated as executable input by CFG construction or later analyses.

\textit{Stack effect annotations.}
Each decoded instruction carries an interval effect
$[\mathit{in}_{\min}, \mathit{in}_{\max}] \to
[\mathit{out}_{\min}, \mathit{out}_{\max}]$, used by continuation
resolution and the detectors. \Cref{alg:disasm} splits the BOC root into the
main code body and an optional method dictionary in \texttt{c3}, then decodes
each entry into a linear instruction sequence.

\begin{algorithm}[t]
\LinesNumbered
\caption{TVM Bytecode Disassembly}\label{alg:disasm}
\KwIn{BOC bytecode $B$}
\KwOut{Method-indexed instruction sequences $\mathcal{I}$}
$(\mathit{code}, \mathit{dict}) \gets \mathsf{SplitCodeAndDict}(\mathsf{ParseBOC}(B).root)$\;
$\mathcal{I} \gets \{0 \mapsto \mathsf{DecodeSlice}(\mathit{code})\}$\;
\If{$\mathit{dict} \neq \bot$}{
  \ForAll{$(\mathit{key}, \mathit{val}) \in \mathsf{ParseDict}(\mathit{dict})$}{
    $\mathcal{I}[\mathit{key}] \gets \mathsf{DecodeSlice}(\mathit{val})$\;
  }
}
\Return{$\mathcal{I}$}\;
\end{algorithm}

\subsection{Continuation Resolution}\label{sec:cont-resolve}

Given the instruction sequence produced by the disassembler, this stage resolves
which continuation each branch, call, or jump instruction targets.
Since \tvm control-transfer instructions consume continuations from the stack,
determining the target requires static reasoning about stack contents. We
perform lightweight stack simulation using a three-category lattice (\cref{def:stack-token}).

\noindent\textit{Example.}
Consider the two control-transfer sites in \cref{fig:motivating}.
At the blue \texttt{CALLX} site, the preceding \texttt{PUSHCONT} pushes
Cont.1. The stack therefore holds the known continuation token
$\mathsf{cont}(\mathit{Cont.1})$, and the resolver yields an \emph{exact}
target.
By contrast, at the \texttt{JMPX} site (amber), the operand comes from
\texttt{PUSHCTR\;c0}, whose value was restored via the \savelist after the
callee returned. The register map yields $\mathsf{unknown}$ for
$c_0$ under a purely local simulation, leaving the target \emph{unresolved}
until \savelist propagation (\cref{sec:savelist}) recovers the binding.
This gap motivates the intercontinuation fixed-point described below.

\subsubsection{State Representation}

\begin{definition}[Stack Token]\label{def:stack-token}

  A token $t$ is one of:
  (1)~$\mathsf{cont}(id)$ --- a known continuation with identifier $id$;
  (2)~$\mathsf{val}(n?)$ --- a value (optionally a known constant $n$); or
  (3)~$\mathsf{unknown}$ --- an unresolvable entry.
\end{definition}

\begin{definition}[Analysis Context]\label{def:abs-ctx}

  Let $\mathit{Tok}$ be the set of stack tokens
  (\cref{def:stack-token}). An analysis stack is a pair
  $\sigma = \langle p, u \rangle$ where $p \in \mathit{Tok}^*$ is a finite
  sequence of tokens (the tracked stack prefix) and $u \in \mathbb{B}$ (Booleans) is the
  \textnormal{\texttt{unknown\_below}} flag,
  indicating whether unmodeled elements may exist below that prefix;
  when $u$ is set, accesses beyond $p$ conservatively yield
  $\mathsf{unknown}$. A register map is
  $\rho : \{c_0,c_1,c_2,c_3\} \to \mathit{Tok}_{\bot}$, where
  $\mathit{Tok}_{\bot}$ extends tokens with $\bot$ (undefined); only
  the control-flow registers \textnormal{\texttt{c0}}--\textnormal{\texttt{c3}}
  are tracked, as they
  determine continuation targets. A register read yields $\rho[r]$ when
  defined and $\mathsf{unknown}$ otherwise. An analysis context
  is $\chi = \langle \sigma, \rho \rangle$.
\end{definition}
\subsubsection{Intracontinuation Analysis}

\Cref{alg:cont-resolve} outlines the intracontinuation resolution procedure.
Given an instruction sequence and an initial analysis context $\chi_0$, it
scans each instruction in order, performing three actions per step:
resolving continuation targets from the current state
(Line~\ref{line:resolve}), updating the analysis context
(Line~\ref{line:step}), and emitting transfer edges for intercontinuation
propagation (Line~\ref{line:xfer}). The procedure produces four outputs:
resolved targets $\mathcal{M}$, exact save-target facts $\mathcal{S}$,
unresolved sites $\mathcal{U}$, and callee transfer facts $\mathcal{X}$.

\noindent\textit{Target resolution.}
At each control-transfer instruction, the resolver reads the analysis context
to determine the jump target
(Lines~\ref{line:resolve}--\ref{line:needs}); control-transfer instructions
that yield no target are recorded as unresolved in $\mathcal{U}$.
The resolution uses an
instruction-specific specification: calls and jumps read the stack top; unary
conditionals read the continuation below the condition; multi-target branches
(\eg, \texttt{IFELSE}) read the required $k$ continuation positions.
Each resolved fact carries a confidence tag
$\kappa \in \{\mathit{exact}, \mathit{heur}\}$: facts derived directly from
the analysis context or from exact method-dictionary lookups are
$\mathit{exact}$; recoveries from the dynamic solver are $\mathit{heur}$.
When the current instruction is a save operation and the register read of
$c_0$ yields a known continuation, an exact save-target fact
$\mathcal{S}[i] = id$ is recorded (Line~\ref{line:saveop}).

\noindent\textit{State transfer and propagation.}
The step function (Line~\ref{line:step}) applies exact rules for pushes,
register reads/writes, and stack permutations. For instructions whose stack
behavior is not modeled exactly, the interval stack-effect annotation
$[\mathit{in}_{\min},\allowbreak \mathit{in}_{\max}] \to
[\mathit{out}_{\min},\allowbreak \mathit{out}_{\max}]$ from \cref{sec:disasm} pops the
known minimum, sets the \texttt{unknown\_below} flag when additional hidden
operands may be consumed, and pushes $\mathsf{unknown}$ outputs.
Finally, typed transfer edges and callee contexts are emitted for
intercontinuation propagation (Line~\ref{line:xfer}).

\begin{algorithm}[t]
\LinesNumbered
\caption{Intracontinuation Resolution}\label{alg:cont-resolve}
\textbf{function} $\mathsf{ResolveLocal}(\mathbf{I}=\langle I_0,\ldots,I_{n-1}\rangle,\; \chi_0)$\;
\KwOut{Control-transfer facts $\mathcal{M}$; exact save-target map $\mathcal{S}$; unresolved set $\mathcal{U}$; transfer facts $\mathcal{X}$}
$\mathcal{M}, \mathcal{S}, \mathcal{U}, \mathcal{X} \gets \emptyset$\;
$\chi \gets \chi_0$\;
\For{$i = 0$ \KwTo $n-1$}{
  $T \gets \mathsf{ResolveTargets}(\chi, I_i)$\nllabel{line:resolve}\;
  \lIf{$\mathsf{NeedsTarget}(I_i) \land T = \emptyset$}{$\mathcal{U} \gets \mathcal{U} \cup \{i\}$}\nllabel{line:needs}
  $\mathcal{M}[i] \mathrel{{+}{=}} T$\;
  \lIf{$\mathsf{SaveOp}(I_i) \land \mathsf{ReadReg}(\chi.\rho, c_0)=\mathsf{cont}(id)$}{$\mathcal{S}[i] \gets id$}\nllabel{line:saveop}
  $\chi' \gets \mathsf{Step}(\chi, I_i)$\nllabel{line:step}\;
  $\mathcal{X} \mathrel{{+}{=}} \mathsf{TransferEdges}(\chi, \chi', I_i, T)$\nllabel{line:xfer}\;
  $\chi \gets \chi'$\;
}
\Return{$\mathcal{M}, \mathcal{S}, \mathcal{U}, \mathcal{X}$}
\end{algorithm}

\textit{Known-Prefix Fallback.}
When a control-transfer instruction cannot resolve its continuation via
positional access (because the target depth exceeds the known prefix), a
fallback scans the entire known prefix for $\mathsf{cont}$ tokens. If exactly
$k$ continuations are found and the instruction requires exactly $k$ targets,
they are returned as the resolution. An exact-count guard prevents false
matches when extra continuations are present.

\subsubsection{Intercontinuation Fixed-Point}

The intracontinuation resolution of \cref{alg:cont-resolve} analyzes one
continuation body in isolation, so callees start with no knowledge of the
caller's stack or registers. To recover caller-derived state, we introduce an
intercontinuation fixed-point worklist (\cref{alg:interproc}) that repeatedly
invokes \cref{alg:cont-resolve} while propagating entry contexts across
continuation boundaries until convergence.

\begin{algorithm}[t]
\LinesNumbered
\caption{Intercontinuation Fixed-Point Resolution}\label{alg:interproc}
\KwIn{Continuation contexts $C$; entry points $E$}
\KwOut{Control-transfer facts $\mathcal{M}$; exact save-target map $\mathcal{S}$; unresolved set $\mathcal{U}$}
$\mathit{ctx}[c] \gets \bot$ for all $c \in C$\tcp*{Analysis contexts}
$\mathit{out}[c] \gets (\emptyset,\emptyset,\emptyset,\emptyset)$ for all $c \in C$\;
\ForAll{$e \in E$}{
  $\mathit{ctx}[e] \gets \mathsf{InitCtx}(e)$\nllabel{line:initctx}\;
}
$\mathcal{W} \gets E$\;
\While{$\mathcal{W} \neq \emptyset$}{
  $c \gets \mathcal{W}.\mathrm{dequeue}()$\;
  $(\mathcal{M}_c, \mathcal{S}_c, \mathcal{U}_c, \mathcal{X}_c)
    \gets \mathsf{ResolveLocal}(\mathsf{Body}(c), \mathit{ctx}[c])$\nllabel{line:resolvelocal}\;
  $\mathit{out}[c] \gets (\mathcal{M}_c, \mathcal{S}_c, \mathcal{U}_c, \mathcal{X}_c)$\;
  \ForAll{$(c', \chi_t) \in \mathcal{X}_c$}{
    $\mathit{ctx}' \gets \mathit{ctx}[c'] \sqcup \chi_t$\nllabel{line:join}\;
    \If{$\mathit{ctx}' \neq \mathit{ctx}[c']$}{
      $\mathit{ctx}[c'] \gets \mathit{ctx}'$\;
      $\mathcal{W}.\mathrm{enqueue}(c')$\nllabel{line:enqueue}\;
    }
  }
}
$(\mathcal{M}, \mathcal{S}, \mathcal{U}) \gets \mathsf{Aggregate}(\mathit{out})$\nllabel{line:aggregate}\;
$(\mathcal{M}, \mathcal{S}, \mathcal{U}) \gets
  \mathsf{PruneUnreachable}(\mathcal{M}, \mathcal{S}, \mathcal{U}, E)$\nllabel{line:prune}\;
\Return{$\mathcal{M}, \mathcal{S}, \mathcal{U}$}
\end{algorithm}

\noindent Each public entry point is seeded with an initial analysis context
(Line~\ref{line:initctx}): the standard \ton stack shape for the main entry,
or the conservatively widened shape observed at method-dictionary dispatch
sites. The main loop dequeues each continuation and invokes the intracontinuation
resolution of \cref{alg:cont-resolve} (Line~\ref{line:resolvelocal}).
The propagated callee context is merged into the target's entry state via
the join operator $\sqcup$
(Lines~\ref{line:join}--\ref{line:enqueue}). The join widens stack prefixes
element-wise (mismatched positions become $\mathsf{unknown}$) and applies
\emph{unanimity merging} to registers: a register value is retained only when
all callers agree, and discarded otherwise. This correctly propagates globally
unique bindings such as the method dictionary $c_3$ while preventing false
propagation of caller-specific values such as return addresses in $c_0$.
If the merged entry state changes, the target is re-enqueued, driving
iteration to a fixed point.
Termination follows because the tracked stack prefix, token categories,
register domain, and continuation set are finite; $\sqcup$ is monotone and can
only replace information with $\mathsf{unknown}$ or discard a non-unanimous
register value. Thus each entry context changes finitely many times.
After convergence, the final intracontinuation facts are aggregated
(Line~\ref{line:aggregate}), and a depth-first reachability traversal from
entry points prunes unresolved edges originating from unreachable code
(Line~\ref{line:prune}), as these represent dead continuations extracted from
data cell references (\eg, \texttt{PUSHREFSLICE}) rather than executable code.
\texttt{SAVE}-family instructions whose pre-save $c_0$ binding is not exact
are omitted from $\mathcal{S}$ and hence from the savelist construction of
\cref{sec:savelist}.

\textit{Dynamic Target Solver.}
When continuation resolution leaves a target unresolved, a fallback solver tries
four common patterns before emitting $\mathsf{dynamic}$: direct dictionary
dispatch from a preceding \texttt{PUSHINT}, exact \texttt{PUSHCTR\,c3; EXECUTE}
recovery, a weaker method-ID hint from \texttt{PUSHINT}, and
\texttt{BLESS; EXECUTE}. Only the exact \texttt{c3}-based recovery contributes
to save-target attribution and \cref{thm:savelist-sound}; heuristic recoveries
are used only to improve later CFG coverage. Each call/jump site is then
annotated with the target's \savelist summary rather than recursively inlining
callees.

\subsection{TASIR: Intermediate Representation}\label{sec:ir}

Raw \tvm instructions expose numerous encoding variants with identical
semantics. Writing analyses directly over these encodings would duplicate
transfer functions and detector rules. Our lifter therefore normalizes them
into semantic node kinds in \tasir, a typed intermediate representation, while
retaining the original \tvm instruction in each node for traceability. Using
the outputs of \cref{alg:interproc}, $\mathcal{M}$ drives the construction of a
typed CFG that can be reused by downstream analyses, while $\mathcal{S}$ seeds
the savelist modeling described in \cref{sec:savelist}.

Each \tasir instruction carries the original opcode and a \emph{kind} that
classifies its behavior, such as \textsf{CONT\_CALL}, \textsf{STACK\_PUSH},
or \textsf{CELL\_LOAD}. It also carries the interval stack effect from
\cref{sec:disasm} and optional security labels. Examples include
\textsf{MESSAGE\_SEND} and \textsf{AUTH\_CHECK}.
The kind field drives transfer-function dispatch in the taint analysis
(\cref{sec:taint}), and labels identify sinks and sources.
Because each node retains its original instruction and exposes an ordinary
typed CFG, \tasir can also support analyses beyond the detectors in this paper.

\begin{definition}[\tasir Module]\label{def:tasir-module}

  A \tasir module is a tuple
  $\Pi = \langle F, D, \mathit{Edges} \rangle$ where $F$ is the set of
  per-continuation functions, $D$ maps continuation identifiers to
  descriptors, and $\mathit{Edges}$ is the set of typed cross-function edges
  (call, jump, return, exception). Each continuation descriptor
  $d \in D$ records the continuation's identifier, entry block, basic blocks,
  and a \savelist summary~$\hat{\ell}$ populated by the savelist
  construction of \cref{sec:savelist}.
\end{definition}

The \emph{CFG Builder} consumes $\mathcal{M}$ and constructs a CFG with typed edges (fallthrough,
branch, call, call\_return); call-return edges are emitted only when the
$c_0$ binding is known exactly, as the transfer facts already carry these
bindings. \tasir organizes the resulting basic blocks into per-continuation
functions and attaches each function to its descriptor in $D$.
The \emph{Stack Analyzer} then computes per-block stack height bounds, and the
\emph{Event Extractor} classifies security-relevant effects (\eg, message
sends, gas commitment, persistent storage writes) for the detectors in
\cref{sec:detectors}.
This hierarchy makes the relevant objects
explicit: the savelist construction (\cref{sec:savelist}) iterates over
functions $f \in F$ for forward analysis and over continuation descriptors
$d \in D$ when attaching savelist summaries and propagating them to call
sites.

\subsection{Static Savelist Modeling}\label{sec:savelist}

Continuation resolution (\cref{sec:cont-resolve}) determines \emph{where}
control flows, while savelist construction determines \emph{what data} crosses
those continuation boundaries. We keep them separate because savelist analysis
depends on resolved edges but uses a stronger register-domain abstraction.

\subsubsection{Formal Definitions}

\begin{definition}[Tracked Value]\label{def:absval}

  A tracked value is a tuple
  \[
    \hat{v} = \langle \mathit{src} : \mathit{Source},\;
    \mathit{def} : \mathbb{N} \rangle
  \]
  where $\mathit{src}$ is the provenance
  (\eg, \textsf{STORAGE}, \textsf{MSG\_SENDER}) and
  $\mathit{def}$ is the definition site (instruction index) of a
  \emph{locally tracked} register write within the currently analyzed
  continuation body.
\end{definition}

\begin{definition}[Register State]\label{def:abs-reg-state}

  Let $\hat{V}$ denote the set of all tracked values
  (\cref{def:absval}). A register state is a total function
  $\hat{\rho} : \mathit{Reg} \to \mathcal{P}(\hat{V})$,
  where $\mathcal{P}$ denotes the power set.
  For a register $r$, the set $\hat{\rho}(r)$ contains the locally tracked
  definition sites that may reach $r$ at the current program point.
  $\hat{\rho}(r)=\emptyset$ means that the current value of $r$ is either
  not locally defined on the explored path or has been invalidated by an
  analysis barrier (\eg, a dynamic-register write).
\end{definition}

\begin{definition}[Savelist Summary]\label{def:abs-savelist}

  A \savelist summary is a pair
  $\hat{\ell} = \langle \mathit{saved}, \mathit{defs} \rangle$
  where $\mathit{saved} \subseteq \mathit{Reg}$ records which registers may be
  explicitly saved for the continuation, and
  $\mathit{defs} : \mathit{Reg} \to \mathcal{P}(\hat{V})$ maps each register
  to the set of locally tracked definition sites that may flow into that
  saved entry. For a save of register $r$ with reaching-definition set
  $\hat{R}_r$, the update is:
  $\mathit{saved}' = \mathit{saved} \cup \{r\}$ and
  $\mathit{defs}'(r) = \mathit{defs}(r) \cup \hat{R}_r$.
\end{definition}

\subsubsection{Construction Algorithm}

\Cref{alg:savelist} formalizes the construction in three phases, operating
over the tracked values (\cref{def:absval}), register states
(\cref{def:abs-reg-state}), and savelist summaries
(\cref{def:abs-savelist}).
The first phase (Lines~\ref{line:sp}--\ref{line:sp-end}) collects all
save-instruction sites into $\mathit{SP}$, where $I.idx$ denotes the
position index of instruction $I$: for each save instruction with statically
known saved registers, the saved-register set is extracted
and the exact save-target map $\mathcal{S}$ supplies the continuation bound
to $c_0$ at that site. The second phase (Line~\ref{line:fwd}) computes
$\Phi$, the least fixed point of reaching definitions over the
register-state lattice of \cref{def:abs-reg-state}: modeled register writes
attach provenance, non-write instructions propagate states conservatively,
and CFG joins use pointwise set union.
Lines~\ref{line:build}--\ref{line:build-end} build each continuation's
\savelist summary by recording both the registers that may be saved and the
tracked definitions that may populate those saved entries.

\begin{algorithm}[t]
\LinesNumbered
\caption{Savelist Construction}\label{alg:savelist}
\KwIn{\tasir module $\Pi=\langle F,D,\mathit{Edges}\rangle$; exact save-target map $\mathcal{S}$}
\KwOut{Enhanced module $\Pi'$ with populated savelists}

$\mathit{SP} \gets \emptyset$\nllabel{line:sp}\tcp*{Save points}
\ForAll{instruction $I$ in functions of $F$}{
  \If{$\mathsf{SaveOp}(I) \land I.idx \in \mathrm{dom}(\mathcal{S})$}{
    $R_s \gets \mathsf{SavedRegs}(I)$\;
    $t \gets \mathcal{S}[I.idx]$\;
    $\mathit{SP}[t] \mathrel{{+}{=}} \{(I.idx, R_s)\}$\;
  }
}\nllabel{line:sp-end}

$\Phi \gets \mathsf{ReachDefs}(\Pi)$\nllabel{line:fwd}\tcp*{Reaching-definition fixed point}

\tcp{Savelist build}\nllabel{line:build}
\ForAll{continuation descriptor $d \in D$}{
  $d.\hat{\ell}.saved \gets \emptyset$\;
  \ForAll{$r \in \mathit{Reg}$}{
    $d.\hat{\ell}.defs[r] \gets \emptyset$\;
  }
  \ForAll{$(i_s, R_s) \in \mathit{SP}[d.id]$}{
    \ForAll{$r \in R_s$}{
      $d.\hat{\ell}.saved \gets
        d.\hat{\ell}.saved \cup \{r\}$\;
      $d.\hat{\ell}.defs[r] \gets
        d.\hat{\ell}.defs[r] \cup \Phi[i_s][r]$\;
    }
  }
}\nllabel{line:build-end}

\Return{$\Pi'$}
\end{algorithm}

\noindent\textit{Example (\cref{fig:motivating}).}
\Cref{tab:savelist-example} traces \cref{alg:savelist} on the motivating
example: \texttt{CALLX} in Cont.0 triggers a save of \texttt{c0} into
Cont.1's \savelist (Phase~1), reaching-definition analysis identifies
the saved value as the return-point continuation (Phase~2), and the
\savelist summary is built accordingly (Phase~3).
With this summary, the previously unresolved \texttt{PUSHCTR\;c0; JMPX}
(amber in \cref{fig:motivating}) can be resolved: the target is the
return-point continuation, recovered via the constructed savelist.
The red dashed path in \cref{fig:motivating} further shows how taint
from \texttt{LDMSGADDR} propagates through \texttt{c4} via this savelist
mechanism to reach \texttt{SENDRAWMSG} in Cont.1.

\begin{table}[t]
  \caption{Savelist construction walkthrough for \cref{fig:motivating}.}
  \label{tab:savelist-example}
  \centering
  \begin{tabularx}{\columnwidth}{@{}cllX@{}}
    \toprule
    \textbf{Ph.} & \textbf{Action} & \textbf{Input} & \textbf{Output} \\
    \midrule
    1 & Collect save & \texttt{CALLX} at $i_s$ & $\mathit{SP}[\text{C1}]\!=\!\{(i_s,\{c_0\})\}$ \\
    2 & ReachDefs & $\Phi[i_s][c_0]$ & $\{\text{ret-point}\}$ \\
    3 & Build $\hat{\ell}$ & C1 descriptor & $\mathit{saved}\!=\!\{c_0\}$, $\mathit{defs}[c_0]\!=\!\{\text{ret-pt}\}$ \\
    \bottomrule
  \end{tabularx}
\end{table}

\subsubsection{Soundness Scope}

For a concrete value $v$, let $\mathrm{def}(v)$ denote the instruction index
that produced $v$ at runtime, and $\mathrm{dom}(f)$ the domain of a partial
function $f$. For a tracked value $\hat{v}$, let
$\gamma(\hat{v}) = \{ v \in \mathit{Value} \mid
\mathrm{def}(v) = \hat{v}.def \}$; the $\mathit{src}$ tag refines precision
but does not participate in the soundness argument. We extend $\gamma$
pointwise to sets of tracked values.

\begin{theorem}[Tracked-Definition Soundness of Savelist Construction]\label{thm:savelist-sound}

  Let $\hat{\ell}_c = \langle \mathit{saved}_c, \mathit{defs}_c \rangle$ be
  the \savelist summary computed by \cref{alg:savelist} for continuation
  descriptor $c$ (notation follows \cref{def:abs-savelist}).
  Consider any concrete execution trace $\pi$ and any concrete continuation
  instance $\kappa$ of code $c$ whose \savelist is updated in $\pi$ by a
  \textnormal{\texttt{SAVE}}-family instruction whose target continuation is
  \emph{exact-resolved} in the save-target map $\mathcal{S}$ returned by
  \cref{alg:interproc} and whose saved-register set is
  statically known. Let $\ell_\kappa^{\pi}$ be the
  concrete \savelist carried by that instance after those updates.
  Then:
  \begin{enumerate}
    \item $\mathrm{dom}(\ell_\kappa^{\pi}) \subseteq
      \mathit{saved}_c$ \quad (saved-register coverage), and
    \item $\forall\, r \in \mathrm{dom}(\ell_\kappa^{\pi})$:
      if the concrete value $\ell_\kappa^{\pi}(r)$ is \emph{locally tracked}
      at the corresponding save site, then
      $\ell_\kappa^{\pi}(r) \in \gamma\bigl(\mathit{defs}_c(r)\bigr)$
      \hfill(tracked-value over-approximation).
  \end{enumerate}
  A value is locally tracked when its last reaching write before the save site
  is a modeled $\mathsf{RegWrite}$ in the analyzed function and no subsequent
  analysis barrier invalidates that register. Save sites with heuristic or
  unresolved continuation targets, or with dynamic saved-register operands, are
  conservatively excluded from the precise value-attribution part of the
  theorem.
\end{theorem}

\noindent\emph{Composition with continuation resolution.}
The theorem relies on \cref{alg:cont-resolve} for the exact save-target fact
$\mathcal{S}[i_s] = c$ and the statically known saved-register set at $i_s$.
Heuristic targets, dynamic targets, and dynamic saved-register operands are
excluded from the precise value-attribution claim and remain only in the
implementation's conservative fallback.

\noindent\textit{Proof sketch.}
Let $L$ be the register-state lattice of
\cref{def:abs-reg-state}, ordered pointwise by subset. Because the module and
$\hat{V}$ are finite, $L$ is finite. The \texttt{ReachDefs} transfer is
monotone: modeled writes install singleton definition sets, the non-write
transfer preserves or clears information conservatively, and joins use
pointwise union. Thus \texttt{ReachDefs} computes the least fixed point for
locally tracked register writes.

For any reachable save site $i_s$ and register $r$, $\Phi[i_s][r]$ contains
every locally tracked definition that may reach $r$ immediately before the save.
For Property~1: by the theorem's precondition the saved-register set is
statically known, so the concrete registers saved at runtime equal the register
set encoded by the save instruction at $i_s$. Lines~\ref{line:sp}--\ref{line:sp-end} record
exactly these registers into $\mathit{SP}$, and
Lines~\ref{line:build}--\ref{line:build-end} union them into
$\mathit{saved}_c$; hence
$\mathrm{dom}(\ell_\kappa^{\pi}) \subseteq \mathit{saved}_c$.
For Property~2: when the concrete value of $r$ is locally tracked (i.e., its
last write is a modeled $\mathsf{RegWrite}$ not invalidated by an analysis
barrier), $\Phi[i_s][r]$ includes that definition by the least fixed-point
guarantee. Lines~\ref{line:build}--\ref{line:build-end} union the
reaching-definition sets from all save sites targeting $c$ into
$\mathit{defs}_c(r)$, so the concrete value is in
$\gamma(\mathit{defs}_c(r))$.

\subsection{Path-Sensitive Taint Analysis}\label{sec:taint}

With the \tasir module and populated savelists, this stage tracks how
attacker-controlled data propagates through the program, including across
continuation boundaries via savelist-restored registers.

\subsubsection{Taint State}\label{sec:taint-domain}

\begin{definition}[Taint Lattice]\label{def:taint-lattice}

  The taint lattice is:
  \[
    \hat{d} = \langle \mathit{src} : \mathit{Source},\;
    \mathit{tainted} : \mathbb{B},\;
    \mathit{checked} : \mathbb{B},\;
    \mathit{origins} : \mathcal{P}(\mathbb{N}) \rangle
  \]
  where $\mathit{origins}$ tracks the definition sites from which
  taint propagated, and $\mathit{Source}$ is one of:
  \textsf{MSG\_SENDER}, \textsf{MSG\_BODY}, \textsf{MSG\_VALUE},
  \textsf{CONSTANT}, \textsf{COMPUTATION}, \textsf{STORAGE}, or
  \textsf{UNKNOWN}.
\end{definition}

Source tags form a flat order (all tags are incomparable; their join is \textsf{UNKNOWN}). Taint values are
ordered componentwise: source tags use that flat order, $\mathit{tainted}$ uses
the usual Boolean order, $\mathit{checked}$ uses the reverse Boolean order (so
unchecked is higher), and $\mathit{origins}$ uses subset. Merge points keep the
source tag when both sides agree and use \textsf{UNKNOWN} otherwise; they OR
the taint bit, AND the checked bit, and union the origin sets.

The analysis state $\hat{S}$ stores four components: a bounded stack of up to
$K_s$ taint values, an $\mathit{overflow}$ bit, a register map from
$\mathit{Reg}$ to $\hat{d}_\bot$, and a partial guard map from definition sites
to guard instructions. Only the top $K_s$ stack entries are tracked
explicitly; deeper pushes discard the deepest tracked entry and set
$\mathit{overflow} = \mathit{true}$.
$K_s$ is a fixed implementation parameter.

\subsubsection{Transfer Functions}\label{sec:transfer}

We define four categories of taint transfer rules (\cref{tab:transfer}).

\begin{table}[t]
  \caption{Taint transfer function categories.}
  \label{tab:transfer}
  \centering
  \begin{tabularx}{\columnwidth}{@{}p{0.22\columnwidth}p{0.25\columnwidth}X@{}}
    \toprule
    \textbf{Category} & \textbf{Kind} & \textbf{Rule} \\
    \midrule
    Uncond.\ source
      & Msg-field read
      & $\hat{d}_{out}.tainted \gets \mathit{true}$ \\[2pt]
    Cond.\ source
      & Generic load
      & Taint only if input is msg-derived \\[2pt]
    Propagation
      & Data transform
      & $\hat{d}_{out} \gets \hat{d}_{in}$ \\[2pt]
    Sink
      & Sensitive op
      & Report if tainted $\land$ unchecked \\
    \bottomrule
  \end{tabularx}
\end{table}

\emph{Unconditional sources} always produce tainted output because they read
attacker-controlled message fields (sender address, value, body).
\emph{Conditional sources} taint their outputs only when the input originates
from an incoming message, not from persistent storage (\texttt{c4}); this
distinction prevents over-tainting of contract-internal data.
\emph{Propagation} instructions inherit the input taint unchanged; for
multi-output instructions, each output applies its own inheritance rule.
\emph{Sinks} do not transform taint but trigger a finding when an operand is
tainted and unchecked.

\textit{Guard analysis.}\label{sec:guards}
When a guard instruction (\texttt{IF}, \texttt{THROWIF}, etc.)
is encountered and its condition depends on a tainted value,
the corresponding $\mathit{checked}$ bit is set, suppressing
further reports for that value at downstream sinks.

\subsubsection{Cross-Continuation Taint via Savelist}\label{sec:cross-taint}

Without savelist-aware propagation, taint information is lost at continuation
boundaries (\cref{fig:motivating}), causing the analysis to miss
vulnerabilities that involve data flowing through saved registers.
To bridge this gap, at each continuation call/jump instruction $I$,
\tool retrieves the precomputed savelist summary (\cref{sec:savelist})
$\hat{\ell} = \langle \mathit{saved}, \mathit{defs} \rangle$. For each saved register
$r \in \mathit{saved}$ and each locally tracked definition $\hat{v} \in \mathit{defs}(r)$,
if $\hat{v}$ is tainted in the current state, that taint is propagated from
definition site $\hat{v}.def$ to the call/jump site $I.idx$:

\begin{equation}\label{eq:cross-taint}
  \begin{aligned}
    &\forall r \in \mathit{saved},\; \forall \hat{v} \in \mathit{defs}(r):\;
      \mathit{tainted}(\hat{v}.def) \\
    &\implies \mathit{propagate}(\hat{v}.def, I.idx)
  \end{aligned}
\end{equation}
\noindent Here $\mathit{tainted}(i)$ is shorthand for ``the current taint
state contains a tainted value whose origin/definition site is $i$,'' and
$\mathit{propagate}(i,j)$ denotes the cross-continuation taint fact forwarded
from definition site $i$ to call/jump site $j$.

Upon entering the target continuation, the \savelist entries are restored into
the register state. The restore-side transfer function updates the taint
analysis state accordingly:
\begingroup
\begin{equation}\label{eq:restore-taint}
  \forall r \in \mathit{saved}:\;
  \hat{S}'.regs[r] \gets
  \mathsf{RestoreSet}(\hat{S}, r, \mathit{defs}(r))
\end{equation}
\endgroup
\noindent where the restore helper joins the taint states for all tracked
definitions in $\mathit{defs}(r)$; when no tracked definitions are available
for $r$, it materializes a placeholder for a saved-but-unattributed register.
Thus \cref{eq:cross-taint,eq:restore-taint} form the cross-continuation bridge
$\mathsf{XContTaint}$ used by \cref{alg:worklist}.

\subsubsection{Worklist Algorithm}\label{sec:worklist}

The path-sensitive analysis is driven by the priority worklist of
\cref{alg:worklist}. Public entries are seeded with the standard \ton entry
stack from \cref{sec:bg-ton}: selector as an untainted constant,
message body/message cell/message value as message-derived sources, contract
balance as an untainted environment value, and all control registers at~$\bot$.

\begin{algorithm}[t]
\LinesNumbered
\caption{Path-Sensitive Taint Analysis}\label{alg:worklist}
\KwIn{\tasir module $\Pi'$ with savelists; initial entry pairs
  $\mathcal{E}_0 = \{(B, \hat{S})\}$ where $B$ is a basic block and $\hat{S}$ its taint state}
\KwOut{Taint results $\mathcal{T}$}
$\mathcal{T}, \mathcal{V} \gets \emptyset$\tcp*{Findings; visited (block, fingerprint) pairs}
$\mathcal{W} \gets \mathsf{PriorityQueue}(\mathcal{E}_0)$\tcp*{Seed with entry blocks}
\While{$\mathcal{W} \neq \emptyset$}{
  $(B, \hat{S}) \gets \mathcal{W}.\mathrm{dequeue}()$\;
  $fp \gets \mathsf{Fingerprint}_{K_s}(\hat{S})$\nllabel{line:fp}\;
  \lIf{$(B.\mathrm{id}, fp) \in \mathcal{V}$}{\textbf{continue}}\nllabel{line:visited}
  $\mathcal{V} \gets \mathcal{V} \cup \{(B.\mathrm{id}, fp)\}$\;
  \ForAll{instruction $I$ in $B$}{
    $\hat{S} \gets \mathsf{TaintStep}(I, \hat{S})$\nllabel{line:transfer}\tcp*{Transfer function}
    \If{$I.kind \in \{\textsf{CONT\_CALL}, \textsf{CONT\_JUMP}\}$}{
      $\hat{S} \gets \mathsf{XContTaint}(\hat{S}, I, I.\hat{\ell})$\nllabel{line:xcont}\tcp*{Target's savelist}
    }
    \If{$\mathsf{Sink}(I) \land \mathsf{Unchecked}(\hat{S}, I)$}{
      $\mathcal{T} \gets \mathcal{T} \cup \{(I, \hat{S})\}$\nllabel{line:sink}\;
    }
  }
  \ForAll{successor $B'$ of $B$}{
    \uIf{$\mathsf{BackEdge}(B, B') \land \mathit{loops}(B') > \textsc{MaxUnroll}$}{
      $\mathcal{T} \gets \mathcal{T} \cup \mathsf{LoopSummary}(B', \hat{S})$\nllabel{line:loopsummary}\;
    }
    \Else{
      $\mathcal{W}.\mathrm{enqueue}(B', \hat{S},
        \mathrm{priority}(B'))$\;
    }
  }
}
\Return{$\mathcal{T}$}
\end{algorithm}

Before processing a block, the analysis fingerprints the taint state and
skips blocks already visited with the same fingerprint
(Lines~\ref{line:fp}--\ref{line:visited}).
The fingerprint hashes the bounded stack slice, the $\mathit{overflow}$ bit,
and each tracked taint tuple
$(\mathit{src}, \mathit{tainted}, \mathit{checked}, \mathit{origins})$,
preserving checked-state and provenance distinctions.
Before pruning a repeated state, the implementation confirms structural
equality; no fingerprint collision was observed in our evaluation.
Each instruction is processed by its transfer function
$\mathsf{TaintStep}$ (Line~\ref{line:transfer}). At call/jump sites, the
cross-continuation taint bridge (\cref{sec:cross-taint}) propagates taint
through the target's savelist summary (Line~\ref{line:xcont}).
A finding is recorded when a sink operand is tainted and unchecked
(Line~\ref{line:sink}).
When a successor edge leads back to an already-visited block (a back edge,
indicating a loop) and the unroll bound \textsc{MaxUnroll} is exceeded, a
conservative loop summary reports any tainted values reaching sinks within
the loop body without further unrolling (Line~\ref{line:loopsummary}).
The priority function favors blocks containing sensitive operations and
those whose stack carries unguarded tainted values, directing exploration
toward the most vulnerability-prone paths first.

\textit{Termination.}
The analysis terminates because the state space is finite. Each stack/register
slot ranges over a finite lattice, and only the top $K_s$ stack entries
are tracked. Combined with the finite number of basic blocks and the loop bound
\textsc{MaxUnroll}, the visited set $\mathcal{V}$ is finite and the worklist
eventually empties.

\subsection{Security Detectors}\label{sec:detectors}

\tool implements detectors for 5 of the 8 defect classes defined by
TONScanner~\cite{song2025tonscanner}, operating on \tasir semantic labels
and savelist-enhanced taint facts.
Each detector is described by a
\emph{vulnerability definition} followed by its \emph{detection rule}.
The rules are written in terms of \tasir-level and helper predicates rather
than concrete opcode patterns or other low-level matching heuristics.

\textbf{V1 --- Unhandled Bounced Message.}
\emph{Definition.}
In \ton's asynchronous messaging, a bounced message re-enters
\texttt{recv\_internal}. If the contract does not check the bounced flag,
it processes the bounce as a legitimate message, causing permanent fund loss.

\noindent\emph{Rule.}
{
The detector searches an entry prefix of \texttt{recv\_internal} for
bounced-flag extraction sites that are followed by a guard. The
implementation recognizes several equivalent bit-extraction idioms. Let
$\mathit{Prefix}(f)$ denote the scanned entry prefix of function~$f$,
$\mathit{Bounce}(i)$ hold when instruction~$i$ extracts the bounced flag,
and $G(i)$ hold when instruction~$i$ is followed by a guarding conditional
branch or throw. Formally:
\[
  \not\exists\, i \in \mathit{Prefix}(\texttt{recv\_internal}):\;
  \mathit{Bounce}(i) \land G(i)
\]
}

\textbf{V2 --- Bad Randomness.}
\emph{Definition.}
Predictable blockchain values (\texttt{NOW}, \texttt{BLOCKLT},
\texttt{LTIME}) used as random seeds allow miners or validators to
predict outcomes, enabling front-running or fund manipulation.

\noindent\emph{Rule.}
{
The detector collects randomness-related instructions and
security-sensitive sinks within each function. A finding is reported when a
randomness-to-sensitive pair is found in program order, using this ordering
as a heuristic proxy that predictable values may affect a critical action.
Let $\mathit{Rand}(f)$ and $\mathit{Sens}(f)$ denote the randomness-related
sites and security-sensitive sinks in function~$f$, respectively, and let
$\mathit{Before}(i,j)$ denote program order within a function. Formally:
\[
  \exists\, f \in F,\; r \in \mathit{Rand}(f),\; s \in \mathit{Sens}(f):\;
  \mathit{Before}(r,s)
\]
}

\textbf{V3 --- Lack of end\_parse.}
\emph{Definition.}
A \texttt{Slice} created from a \texttt{Cell} via \texttt{CTOS} is parsed
but never validated with \texttt{ENDS}, potentially accepting malformed
data with hidden payloads.

\noindent\emph{Rule.}
{
For each function, the detector identifies slices that are created from
cells and subsequently read, and compares them against \texttt{ENDS}
validation sites. A finding is reported when at least one parsed slice lacks
a corresponding end-of-slice validation. Let $\mathit{Parsed}(f)$ denote the
parsed slices in function~$f$, and let $\mathit{Ended}(p,f)$ hold when parsed
slice~$p$ is matched by an \texttt{ENDS} validation in function~$f$. Formally:
\[
  \exists\, f \in F,\; p \in \mathit{Parsed}(f):\;
  \neg \mathit{Ended}(p,f)
\]
}

\textbf{V4 --- Precision Loss.}
\emph{Definition.}
An integer division followed by multiplication silently truncates the
intermediate quotient; the fused \texttt{MULDIV} opcode avoids this by
computing the full-width product before dividing.

\noindent\emph{Rule.}
{
The detector looks for a nearby multiplication that follows a division. If
the division site is not a fused \texttt{MULDIV}-family instruction, written
as $\neg \mathit{Fused}(d)$, the pair is flagged as potential precision
loss. Let $\mathit{Div}(f)$ and $\mathit{Mul}(f)$ denote the division and
multiplication sites in function~$f$, and let $\mathit{Near}(i,j)$ hold when
the two sites are within the detector's short matching range. Formally:
\[
  \exists\, f \in F,\; d \in \mathit{Div}(f),\; m \in \mathit{Mul}(f):\;
  \begin{aligned}[t]
    &\mathit{Before}(d,m) \land \mathit{Near}(d,m) \\
    &\land \neg \mathit{Fused}(d)
  \end{aligned}
\]
}

\textbf{V5 --- Improper Modifier.}
\emph{Definition.}
A function lacking the \texttt{impure} modifier that directly or
transitively throws, sends messages, modifies storage, or writes globals
is side-effecting. A call to such a function whose return value is unused
indicates a missing \texttt{impure} annotation.

\noindent\emph{Rule.}
{
Let $\mathit{Call}(f)$ be the call sites in function~$f$,
$\mathit{target}(i)$ the resolved callee of call~$i$,
$\mathit{NoImpure}(g)$ hold when $g$ lacks \texttt{impure},
$\mathit{SE}(g)$ when $g$ is side-effecting, and
$\mathit{DropRet}(i,g)$ when the return value of~$i$ is unused or $g$
returns no value. Formally:
\[
  \exists\, f \in F,\; i \in \mathit{Call}(f),\; g{=}\mathit{target}(i):\;
  \begin{aligned}[t]
    &\mathit{NoImpure}(g) \land \mathit{SE}(g) \\
    &\land\; \mathit{DropRet}(i,\, g)
  \end{aligned}
\]
}


\section{Evaluation}\label{sec:evaluation}

We evaluate \tool with five research questions:

\begin{itemize}
  \item \textbf{RQ1:} How accurate is continuation-target resolution?
  \item \textbf{RQ2:} How does \tool compare with \tsa{} under identical runtime budgets?
  \item \textbf{RQ3:} How effective is \tool on five defect classes?
  \item \textbf{RQ4:} How much does path sensitivity contribute to detection?
  \item \textbf{RQ5:} What is the defect distribution and precision on Registry?
\end{itemize}

\subsection{Experimental Setup}\label{sec:setup}

\textbf{Implementation and Environment.}
\tool is implemented in approximately 17{,}000 lines of Python code
(excluding comments, blank lines, and docstrings), organized into five
modules mirroring the pipeline (\cref{fig:arch}).
BOC deserialization uses \texttt{pytoniq-core}~\cite{pytoniqcore2024}, and
the opcode database is derived from the official \tvm
specification~\cite{tvmspec2023}; all other components are developed from
scratch.
The loop-unrolling bound \textsc{MaxUnroll} is set to~3, and the taint
stack-tracking depth $K_s$ is~64. All experiments run on Apple M4 Max with
16 cores, 128\,GB RAM, and macOS~26.3.
TONScanner is distributed as a pre-built Linux x86-64 ELF binary without
source code; the RQ5 comparison therefore runs inside a Docker container
(Ubuntu~20.04, x86-64 emulation via Rosetta~2).

\textbf{Datasets.}
Two datasets are used: \textbf{Registry}, a
set of 2{,}921 unique \ton contracts collected by us from the \ton verifier
registry~\cite{tonverifier2026} (mainnet and testnet, February 2026), for RQ1,
RQ2, and RQ5. Registry is the deduplicated full snapshot; we apply no
outcome-based filtering. \textbf{Benchmark} contains
208 contracts with published manual labels from
TONScanner~\cite{song2025tonscanner}, for RQ3 and RQ4.
RQ3--RQ5 metrics are computed at the \emph{contract--class pair} level.
The Benchmark contains 190 pairs labeled as vulnerable and
6 pairs labeled as non-vulnerable by TONScanner across the
five evaluated classes.

For baselines, RQ2 compares \tool with \tsa{}~\cite{esprito2024tsa}; RQ3 and
RQ5 compare it
with TONScanner~\cite{song2025tonscanner}. RQ4 evaluates one
component ablation, \tool-NoPS (path sensitivity disabled). All tools use
bytecode-only BOC input and identical hardware. Per-contract timeouts are
1{,}000\,s for RQ2 and 30\,s for RQ4. \tsa{} is run with its default settings.

\subsection{RQ1: Continuation Resolution Completeness and Accuracy}\label{sec:rq1}

\begin{table}[t]
  \caption{Resolution results. All tiers have 100\% target precision and
  exact-set agreement; exact means no missing or spurious oracle targets.}
  \label{tab:rq1}
  \centering
  \begin{tabularx}{\columnwidth}{@{}Xrrl@{}}
    \toprule
    \textbf{Edge Tier} & \textbf{Edges} & \textbf{Ratio}
    & \textbf{Oracle} \\
    \midrule
    Statically determined
      & 284{,}066 & 96.4\% & Structural \\
    Solver-resolved
      & 9{,}452 & 3.2\% & Dictionary \\
    Savelist-dependent
      & 1{,}028 & 0.4\% & Source audit \\
    \midrule
    \textbf{Total}
      & \textbf{294{,}546} & \textbf{100\%} & \\
    \bottomrule
  \end{tabularx}
\end{table}

\begin{table}[t]
  \caption{Cumulative ablation of continuation-resolution techniques.}
  \label{tab:rq1-ablation}
  \centering
  \begin{tabularx}{\columnwidth}{@{}Xrr@{}}
    \toprule
    \textbf{Configuration}
    & \textbf{Unresolved} & \textbf{Resolved} \\
    \midrule
    Baseline (intracont.\ sim.\ + known-prefix)
      & 11{,}081 & --- \\
    \quad + intercont.\ propagation
      & 10{,}612 & 469 \\
    \quad + reachability pruning
      & 6{,}980  & 3{,}632 \\
    \quad + \savelist propagation
      & \textbf{5{,}952}  & \textbf{1{,}028} \\
    \bottomrule
  \end{tabularx}
\end{table}

RQ1 evaluates target resolution along two dimensions. \emph{Accuracy} asks
whether resolved targets are correct; \emph{completeness} asks how many edges
each technique resolves.
Since all 2{,}921 contracts originate from the \ton verifier registry with
verified source code, we construct a deterministic verification oracle for each
resolved edge.

\noindent\textbf{Oracle construction.}
We assign each dynamic edge to one of three tiers based on its resolution
mechanism (\cref{tab:rq1}).
\textbf{Statically determined} edges (284{,}066, 96.4\%) are conditional
branches such as \texttt{IFELSE}, \texttt{IF}, \texttt{IFNOT}, and
\texttt{IFJMP}
whose targets are the \texttt{PUSHCONT} operands immediately preceding
the branch instruction.
Ground truth is extracted structurally from the bytecode: the target
continuation identifiers are uniquely determined by instruction position,
making verification fully automatic and provably correct.
\textbf{Solver-resolved} edges (9{,}452, 3.2\%) are dictionary-dispatch
edges (\texttt{CALLDICT}, \texttt{DICTI\-GET\-JMPZ}) resolved by reading
compile-time dictionary content from~\texttt{c3}.
Ground truth is the set of dictionary entries; we cross-verify these
against source-level function definitions, confirming that the method
identifiers match the declared functions.
\textbf{Savelist-dependent} edges (1{,}028, 0.4\%) are indirect register
calls (\texttt{EXECUTE}/ \texttt{JMPX} via \texttt{PUSHCTR}) whose
targets depend on savelist propagation.
These edges are verified by source-code audit of the 129 affected
contracts: for each edge, we trace the register-flow chain through the
source to confirm that the resolved target matches the call site
semantics. Of the 1{,}028 edges, 1{,}021 (99.3\%) follow the
\mbox{\texttt{PUSHCTR\,c3}} $\to$ \texttt{EXECUTE} pattern; the
remaining 7 edges across 4~contracts are resolved via alternative
solver strategies (\eg, dictionary dispatch). All 1{,}028 edges are
confirmed correct.

\noindent\textbf{Accuracy result.}
For each edge we compare \tool's resolved target set against the
oracle-verified ground truth.
\Cref{tab:rq1} shows that all 294{,}546 edges achieve
\textbf{100\% precision}: the resolved targets match the ground truth
exactly on every edge, verified per-tier as described above.
Per-contract audit reports for the savelist tier are provided in the
artifact package.

\noindent\textbf{Completeness result.}
\Cref{tab:rq1-ablation} shows the cumulative contribution of each
resolution technique, measured by the number of unresolved indirect
call targets (i.e., \texttt{EXECUTE}/\texttt{JMPX} edges with
unknown destinations).
The baseline leaves 11{,}081 targets unresolved.
Intercontinuation propagation resolves 469 targets by propagating
entry-shape information across continuation boundaries.
Reachability pruning eliminates 3{,}632 targets that originate from
unreachable code---dead continuations extracted from data-cell
references (\eg, \texttt{PUSHREFSLICE}) rather than executable
control flow.
Finally, \savelist propagation resolves 1{,}028 targets across
129 contracts (4.4\% of the corpus) by tracking the
\texttt{PUSHCTR\,c0/c3}$\to$\texttt{EXECUTE} pattern at the IR
level.
After all four techniques, 5{,}952 indirect calls remain
unresolved---these arise from continuations passed through complex
stack manipulations where the analysis cannot trace the target
origin---and are conservatively left as unknown in the call graph.
All 294{,}546 resolved targets achieve \textbf{100\% precision}
(\cref{tab:rq1}), i.e., every resolved target is verified correct.


\subsection{RQ2: Budgeted Exploration vs \tsa{}}\label{sec:rq2}

\begin{table}[t]
  \caption{Budgeted exploration on Registry; times are p50/p95.}
  \label{tab:rq2-path}
  \centering
  \begin{tabularx}{\columnwidth}{@{}lXr@{}}
    \toprule
    \textbf{Tool} & \textbf{Outcomes} & \textbf{Time (s)} \\
    \midrule
    \tool & 2{,}921 success; 0 timeout; 0 crash & 0.24/5.52 \\
    \tsa{} & 2{,}834 success; 12 timeout; 75 crash & 4.10/64.90 \\
    \bottomrule
  \end{tabularx}
\end{table}

RQ2 evaluates whether \tool can analyze the entire Registry reliably and
how its throughput compares with \tsa{}.

\noindent\textbf{Result.}
\tool achieves a \textbf{100\% success rate} with zero crashes and zero
timeouts, while \tsa{} succeeds on 97.0\% of contracts
with 12 timeouts and 75 crashes (\cref{tab:rq2-path}).
\tsa{}'s 1{,}000\,s limit is its default budget; \tool's maximum observed
runtime is 56.6\,s, so this budget does not truncate its runs.
\tool's median analysis time (0.24\,s) is \textbf{17$\times$ faster} than
\tsa{}'s (4.10\,s), and the 95th-percentile gap widens to 12$\times$
(5.52\,s vs.\ 64.90\,s).
This speed advantage stems from \tool's dataflow-analysis approach,
which avoids the path explosion inherent in \tsa{}'s symbolic execution.

\noindent\textbf{Code coverage.}
Beyond throughput, the two tools differ fundamentally in analysis
scope.
\tool computes one whole-CFG dataflow fixed point rather than making one
literal pass: instructions are revisited until states stabilize, with loop
edges bounded to three unrollings. It reached every reachable instruction on
all 2{,}921 contracts within the 1{,}000\,s budget.
\tsa{}, as a symbolic executor, explores concrete paths
one at a time; its internal coverage tracker
reports \emph{transitive instruction coverage}---the fraction of
reachable instructions visited across all explored paths---which
is inherently bounded by the exploration budget and subject to path
explosion.
On the 87 contracts (3.0\%) where \tsa{} crashes or times out,
it produces no analysis results at all, whereas \tool completes
with full coverage.
This difference has a direct downstream impact: a defect detector
can only flag patterns in code it has analyzed, so \tool's
complete coverage eliminates an entire class of false negatives
that arise from incomplete exploration.

\subsection{RQ3: Detection Effectiveness}\label{sec:rq3}
\begin{table}[t]
  \caption{Per-class detection rate on Benchmark.}
  \label{tab:rq3-detection}
  \centering
  \begin{tabular}{lrrr}
    \toprule
    \textbf{Defect Class} & \textbf{Labeled}
    & \textbf{Detected} & \textbf{Det.\ Rate} \\
    \midrule
    Bad Randomness           &  2 &  2 & 100.0\% \\
    Improper Modifier        &  6 &  6 & 100.0\% \\
    Unhandled Bounced Msg    & 56 & 56 & 100.0\% \\
    Lack of end\_parse       & 82 & 82 & 100.0\% \\
    Precision Loss           & 44 & 35 &  79.5\% \\
    \midrule
    \textbf{Overall}         & \textbf{190} & \textbf{181} & \textbf{95.3\%} \\
    \bottomrule
  \end{tabular}
\end{table}

RQ3 evaluates \tool on five defect classes from \textbf{Benchmark}.
Three classes are excluded: Global Variable Redefined and Inconsistent Data
require source-level analysis (variable re-declaration and per-function
cell-layout comparison), and Unchecked Return requires source-level
variable-name tracking---none of which can be reliably performed at the
bytecode level.
For each class, we measure \emph{detection rate}: the fraction of contracts
labeled as vulnerable that \tool also flags.
\tsa{} does not provide high-level defect-class labels; its output reports
low-level execution errors (\eg, \texttt{cell-underflow}) rather than
vulnerability types, so it is not included in this comparison.

\noindent\textbf{Results.}
\tool detects 181 of 190 vulnerable contract--class pairs (95.3\%); all
cases in four classes are detected, and the nine misses are Precision Loss
(\cref{tab:rq3-detection}). It also reports findings for all six pairs labeled
non-vulnerable, which are the false positives analyzed below.
\Cref{tab:rq3-precision} shows 96.8\% overall precision, close to
TONScanner's 96.9\%, without requiring source code.

\begin{table}[t]
  \caption{Contract-level precision comparison.}
  \label{tab:rq3-precision}
  \centering
  \begin{tabularx}{\columnwidth}{@{}l*{4}{>{\centering\arraybackslash}X}@{}}
    \toprule
    & \multicolumn{2}{c}{\textbf{\tool}} & \multicolumn{2}{c}{\textbf{TONScanner}} \\
    \cmidrule(lr){2-3} \cmidrule(lr){4-5}
    \textbf{Language} & \textbf{TP/FP} & \textbf{Precision} & \textbf{TP/FP} & \textbf{Precision} \\
    \midrule
    FunC  & 138\,/\,6 & 95.8\% & 144\,/\,6 & 96.0\% \\
    Tact  & 43\,/\,0  & 100\%  & 46\,/\,0  & 100\%  \\
    \midrule
    \textbf{Combined} & \textbf{181\,/\,6} & \textbf{96.8\%}
                       & \textbf{190\,/\,6} & \textbf{96.9\%} \\
    \bottomrule
  \end{tabularx}
\end{table}

\noindent\textbf{False-positive analysis.}
Six contract--class pairs are over-reported: 4 from Unhandled Bounced Message
and 2 from Lack of end\_parse. The former handle the bounced flag beyond the
detector's entry-prefix window; the latter use parsing patterns not captured
by its opcode heuristic.

\noindent\textbf{False-negative analysis.}
All nine misses are Precision Loss. Recompilation and rescanning confirm that
chained arithmetic separates the vulnerable division and later multiplication
beyond the 12-instruction matching window. This is a rule limitation rather
than a bytecode-availability issue; tracking division results by taint would
remove the fixed-window constraint.

\subsection{RQ4: Path-Sensitivity Ablation}\label{sec:rq4}

Under a 30\,s timeout, the full configuration detects 92.6\% (five contracts
time out), whereas path-insensitive \tool-NoPS detects 91.1\% (four time out).
The remaining 1.5-point gap reflects precision lost by merging path-disjoint
states. Separately, \savelist propagation resolves 1{,}028 additional targets
across 129 contracts (\cref{sec:rq1}). Disabling it does not change the current
detection totals because the five detectors ultimately match local opcode
patterns; its detection-level contribution is therefore a zero-difference
ablation, while its measurable contribution is CFG completeness and support
for future cross-continuation detectors.

\subsection{RQ5: Full-Scale Defect Analysis}\label{sec:rq5}

RQ5 applies \tool to the full Registry to measure defect prevalence and
estimate per-class precision via stratified sampling.

\noindent\textbf{Defect Distribution.}
\tool flags defects in 2{,}806 of 2{,}921 Registry contracts (96.1\%).
Lack of end\_parse is most prevalent (2{,}795; 95.7\%), followed by Unhandled
Bounced Message (1{,}790; 61.3\%); \cref{tab:rq5-findings} gives all classes.

\begin{table}[t]
  \caption{Per-class findings and precision on Registry.}
  \label{tab:rq5-findings}
  \centering
  \begin{tabularx}{\columnwidth}{@{}Xrrr@{}}
    \toprule
    \textbf{Defect Class} & \textbf{Affected (\%)} &
    \textbf{Audit TP/FP} & \textbf{Precision} \\
    \midrule
    Lack of end\_parse     & 2{,}795 (95.7) & 87/6  & 93.5\% \\
    Unhandled Bounce       & 1{,}790 (61.3) & 77/15 & 83.7\% \\
    Bad Randomness         & 304 (10.4)      & 68/6  & 91.9\% \\
    Improper Modifier      & 54 (1.8)        & 29/6  & 82.9\% \\
    Precision Loss         & 854 (29.2)      & 53/19 & 73.6\% \\
    \midrule
    \textbf{Overall}       & \textbf{2{,}806 (96.1)}
    & \textbf{314/52} & \textbf{85.8\%} \\
    \bottomrule
  \end{tabularx}
\end{table}

\noindent\textbf{Precision Sampling.}
We audit a stratified sample (95\% confidence, 10\% margin) against verified
source via independent review by two authors (\cref{tab:rq5-findings}).
Lack of end\_parse is reported under the TONScanner convention that
treats Tact-compiler-omitted \texttt{ENDS} as a true positive (compiler
deficiency), yielding 93.5\%. Unhandled Bounced Message reaches 83.7\%
after review of bounced handlers and send paths. Improper Modifier reaches
82.9\%; its remaining false positives are context-dependent side effects.
Precision Loss is least precise (73.6\%) because its short-window rule lacks
richer value-flow reasoning.

\noindent\textbf{Comparison with TONScanner on Registry.}
We attempt TONScanner on 343 Registry contracts (366 contract--class pairs).
Only 280 are analyzable: 147 are FunC and 133 Tact contracts can be recompiled;
50 Tact projects fail to compile and 13 inputs are non-FunC. Its parser also
rejects some generated FunC, further limiting coverage.
Across the 366 sampled contract--class pairs, TONScanner flags 32 findings
with 84.4\% precision (27~TP, 5~FP), compared with \tool's 85.8\%
precision (314~TP, 52~FP) on the same sample. Thus, the main limitation is
source-language coverage rather than flagged-finding precision.


\section{Related Work}\label{sec:related}

\subsection{Smart Contract Analysis}

Smart contract analysis is well explored for the
\evm~\cite{feist2019slither,tsankov2018securify,mueller2018mythril,luu2016making,brent2020ethainter,bose2022sailfish,jiang2018contractfuzzer,nguyen2020sfuzz,sun2024gptscan},
including formal
semantics~\cite{hildenbrandt2018kevm,schneidewind2020ethor} and
vulnerability taxonomies~\cite{chen2022defining,atzei2017survey}.
For CFG construction from \evm bytecode,
EtherSolve~\cite{contro2021ethersolve} and
EVMLiSA~\cite{arceri2024towards,arceri2025evmlisa} use stack simulation
to resolve \evm jumps. \tool adapts
these stack-simulation ideas to \tvm's continuation model and
extends them with \savelist abstraction, which has no analogue in
\evm.
Beyond Ethereum, WASAI~\cite{chen2022wasai},
VETEOS~\cite{li2024veteos}, and MoveScan~\cite{li2024movescan}
target other blockchains; none handle first-class continuations or
\savelist semantics.

Within \ton, TONScanner~\cite{song2025tonscanner} defines eight
defect classes and operates on the FunC compiler IR, providing
detailed source-level analysis but requiring access to source code;
as shown in our evaluation (\cref{sec:rq5}), its FunC parser rejects
a substantial fraction of our Registry sample due to language-version
incompatibilities.
Misti~\cite{nowarp2024misti} similarly targets \textit{Tact} source
and is limited to contracts written in that language.
\tsa{}~\cite{esprito2024tsa} operates at the bytecode level via
symbolic execution, modeling \savelist operations per
path, but does not build a static summary of
\savelist effects across all paths---each path maintains a concrete
\savelist structure with no join or widening.
Yanovich \etal~\cite{yanovich2025auditrift} derive an audit checklist
from professional \ton audit reports, providing a complementary
knowledge-driven perspective;
BugMagnifier~\cite{yanovich2025bugmagnifier} complements static
analysis with dynamic transaction simulation for runtime validation.
\tool is positioned as a source-free bytecode analyzer with a
dedicated \savelist model, removing the source-code
requirement of TONScanner/Misti, replacing per-path execution
of \tsa{} with a static dataflow model, and handling first-class
continuation semantics natively.

\subsection{Continuation Semantics}

Continuations have a long history in programming-language
theory~\cite{appel1992compiling,felleisen1988theory}, including
continuation-sensitive CFA~\cite{vanhorn2010aam,vardoulakis2011cfa2}
and pushdown
analyses~\cite{vardoulakis2011pushdown,gilray2016p4f} for first-class
control. These techniques target languages where continuations capture
lexical environments; the analysis challenge is representing the
unbounded set of captured bindings. \tvm continuations differ
fundamentally: they carry explicit register snapshots (\savelist{}s)
rather than captured environments, and invocation restores those
registers automatically. The analysis challenge is therefore to
statically approximate which registers are saved at each call site and
what values they hold---a problem closer to interprocedural register
analysis than to environment abstraction. The dedicated savelist
model in \cref{sec:savelist} addresses this challenge. To our
knowledge, \tool provides the first static over-approximation of \tvm
\savelist behavior.

Unlike fixed interprocedural call/return pairs, a first-class continuation
targets the value held in a register and restores an explicit savelist rather
than an implicit frame; analysis must recover both.


\section{Discussion}\label{sec:discussion}

\paragraph{Limitations}
For \textit{bounded analysis}, the soundness guarantee of
\cref{thm:savelist-sound} applies to the \savelist model rather than
the full analysis chain, so the analysis may miss vulnerabilities due to
bounded loop unrolling, path limits, and 5{,}952 indirect calls whose
continuation targets remain unresolved after all resolution techniques
(\cref{sec:rq1}). Such calls remain unknown edges and preserve conservative
taint at the transfer, but their destinations cannot be enumerated; they do not
affect the current local detectors but may hide paths from future
cross-continuation detectors.
The taint-state cache uses structural equality after hashing, preventing hash
collisions from merging unequal states; none occurred in our evaluation.
For \textit{coverage scope}, \tool targets 5 of 8 TONScanner defect classes;
the remaining three require source-level semantics erased during compilation.

\paragraph{Generality and Future Work}
The savelist abstraction is \tvm-specific, but data-driven target resolution
and \tasir's typed CFG generalize to other indirect-jump bytecodes; \evm itself
has no savelists. IFDS would neither remove our bounded operand-stack domain nor
support the non-distributive, path-specific guard state. Future work will add
savelist-aware detectors and solver-assisted indirect-call resolution.

\paragraph{Threats to Validity}
For \textit{internal} validity, ground-truth labels are cross-validated against
TONScanner's published annotations; the RQ1 oracle is exact by construction for
96.4\% of edges and empirically verified for the remainder.
For \textit{external} validity, our dataset contains open-source \ton contracts
only and may not cover all deployment settings.
For \textit{construct} validity, detection-rate and precision estimates are
bounded by TONScanner's label space; manual auditing on a stratified sample
provides cross-validation but cannot fully eliminate labeling bias.


\section{Conclusion}\label{sec:conclusion}

This paper presented \tool, a source-free static analysis framework for \tvm
bytecode whose key contribution is a static dataflow model for \savelist
semantics, enabling taint propagation across continuation boundaries.
Compared with \tsa{}, the only prior bytecode-level analyzer for \tvm,
\tool analyzes all contracts without crashes or timeouts, runs an order of
magnitude faster, and covers every reachable instruction by construction.
Evaluation on 2{,}921 real-world contracts confirms high detection rates
and practical precision across five defect classes, matching the
source-level TONScanner without requiring access to source code.


\section*{Data Availability Statement}

Our replication package is available online at
\url{https://github.com/yxsec/TasmScan_artifact}.

\begin{acks}
This work was supported by the Singapore Ministry of Education Academic
Research Fund Tier 2 (T2EP20224-0003) and Tier 1 (RG12/23), and the Nanyang
Technological University Centre for Computational Technologies in Finance
(NTU-CCTF). The views expressed are those of the authors and not necessarily
those of MOE or NTU-CCTF.
\end{acks}

\clearpage
\balance
\bibliographystyle{ACM-Reference-Format}
\bibliography{references}

\end{document}